\documentclass[smallextended]{svjour3}       
\smartqed  

\usepackage[latin1]{inputenc}
\usepackage[english]{babel}
\usepackage{verbatim}
\usepackage{psfrag}

\usepackage{calc}
\usepackage{graphicx,palatino,verbatim}

\usepackage{amsmath}

\usepackage[T1]{fontenc}

\usepackage{newlfont}

\usepackage{amssymb}
\usepackage{amsmath}

\usepackage{eufrak}

\usepackage{latexsym}

\usepackage{young}
\usepackage[mathscr]{eucal}
\usepackage{amsfonts}
\usepackage{syntonly}
\usepackage{graphicx}
\usepackage{mathptmx}
\usepackage{subfig}
\usepackage{url}
\usepackage{color}
\usepackage[letterpaper]{geometry}

\usepackage{graphicx}

\def\F{\mathbf F}

\def\S{\mathbf S}

\def\deg{\mbox{\rm deg}}

\def\Im{\mbox{\rm Im}}

\def\dim{\mbox{\rm dim}}

\def\FF{{\mathbb F}}

\def\f2{{\mathbb F}_{2}}
\def\fm{{\mathbb F}_{2^m}}

\newcommand{\gd}{\delta}

\newcommand{\gl}{\lambda}

\newcommand{\go}{\omega}

\def\Im{\mathrm{Im}}

\begin{document}

\title{On weak
differential uniformity of vectorial Boolean functions as a cryptographic criterion
}


\author{Riccardo Aragona         \and
        Marco Calderini  \and
        Daniele Maccauro \and
        Massimiliano Sala
}

\institute{F. Author \at
              first address \\
              Tel.: +123-45-678910\\
              Fax: +123-45-678910\\
              \email{fauthor@example.com}           
           \and
           S. Author \at
              second address
}

\institute{R. Aragona \at
              \email{ric.aragona@gmail.com}           
           \and
M. Calderini \at
              \email{marco.calderini@unitn.it}           
           \and
D. Maccauro \at
              \email{daniele.maccauro@gmail.com}           
           \and
M. Sala \at
              \email{maxsalacodes@gmail.com}           
}

\date{Received: date / Accepted: date}

\maketitle

\begin{abstract}
We study the relation among some security parameters for vectorial Boolean functions which prevent attacks on the related block cipher. 
We focus our study on a recently-introduced security criterion, called weak differential uniformity, which prevents the existence of an undetectable trapdoor based 
on imprimitive group action. We present some properties of functions with low weak differential uniformity, especially for the case
of power functions and 4-bit S-Boxes.

\end{abstract}

\section{Introduction}
\label{intro}

Differential and linear attacks are major cryptanalytic tools which apply to most cryptographic algorithms. Therefore, functions which guarantee a high resistance to these attacks have been extensively studied. In particular, those with low differential uniformity and high non-linearity, e.g. Almost Perfect Nonlinear (APN) functions or Almost Bent (AB) functions, have received a lot of attention. Since in the design of a block cipher an invertible S-Box of even dimension is usually needed, there is strong interest in non-linear permutations. However, we know examples of APN permutations in even dimension only for dimension 6, for more details see \cite{apn}.
For the highly interesting cases of dimension $4$ and $8$, the cipher designer will certainly use  $4$-differentially uniform S-Boxes, but she will
also look at other security criteria, if applicable, although it is not obvious which. Besides, even a $4$-differentially uniform S-Box can hide
a trapdoor in the related cipher, if not carefully chosen. Algebraic trapdoors can be very dangerous, especially when they are undetectable \cite{preenel}.

We are investigating the  security criterion introduced recently in  \cite{CVS}, 
called weak differential uniformity. As shown in \cite{CVS}, any cipher (with a prescribed structure) possessing a weakly-APN vectorial Boolean function
as S-Box cannot be successfully attacked using a trapdoor based on imprimitive group action. Indeed, ciphers suffering from such
a trapdoor have been built in \cite{CGC-cry-art-paterson1} and might be used as standards without anyone realizing the trapdoor existence.
So, a designer would have advantage in choosing an invertible S-Box which is simultaneously  weakly APN and $4$-differentially uniform,
which exists for dimension $4$ and $8$ (and actually for any dimension).
Results in \cite{CVS} are generalized on any field in \cite{ACVS}, where again the notion of weakly APN plays an important security role.

In Section 2 we recall the attack \cite{CGC-cry-art-paterson1} that can be mounted on an AES-like cipher when an imprimitive group action
is present. We recall also how a weakly APN S-Box would make this attack ineffective \cite{CVS}, motivating thus this security criterion.
In Section 3 we present some first results on weak differential uniformity. In Section 4 we specialize to the case of monomial functions,
where we see an interesting connection with the property of having the image of a function derivative as an affine space,
which is an unexpected weakness within the underlying algorithms (see for instance \cite{crooked,onan}).
In Section 5 we relate the weak differential uniformity with other algebraic properties of vectorial Boolean functions, in particular with
the degrees of both the function components and the function derivative components.  We can thus improve some results given in \cite{onwAPN} and 
give a formal proof of Fact 4 in \cite{onwAPN}.
In Section 6, we give some results about the partially bent (quadratic), components of a weakly APN permutation and we note that in even dimension weakly APN functions cannot be partially bent (quadratic), behaving thus as APN functions \cite{Ny,SZZ}. In Section 7 we give some other properties of vectorial Boolean functions whose derivatives have no constant components, allowing also a deeper understanding of $4$-bit S-Boxes.

\section{Cryptographic motivations for studying weak differential uniformity}
\label{sec:0}
Most block ciphers used for real-life applications are \emph{iterated block ciphers}, i.e. obtained by a composition of several key-dependent permutations of the message space called ``round functions''. Let  $\mathcal{C}$ be a block cipher, i.e.  a  set of permutations $\{\varphi_k\}_{k\in\mathcal{K}}$  of the message space $V$, where  $\mathcal{K}$ is the key space. An interesting problem is determining the properties of the permutation group $\Gamma_{\infty}(\mathcal{C})=\Gamma_{\infty}$ generated by the round functions of $\mathcal{C}$  that imply  weaknesses of the cipher. 

A property of $\Gamma_{\infty}$ considered undesirable is the imprimitivity. Paterson \cite{CGC-cry-art-paterson1} showed that if this group is imprimitive, then it is possible to embed a trapdoor in the cipher. On the other hand, if the group is primitive no such trapdoor can be inserted. We give the idea of the basic chosen-plaintext attack by Paterson. First we recall  what {it }is an imprimitive group. Let $G$  be a finite group acting  transitively on a set  $V$. We will
write the action of $g\in G$ on an element $v\in V$ as $v g$.
A   \textit{partition}   $\mathcal{B}$   of   $V$  is   said   to   be
$G$-\textit{invariant}     if     $B g\in\mathcal{B}$,    for     every
$B\in\mathcal{B}$   and  $g\in  G$.   A  partition   $\mathcal{B}$  is
\textit{trivial} if  $\mathcal{B}=\{V\}$ or $\mathcal{B}=\{\{v\}\, |\,
v\in V\}$. A non-trivial  $G$-invariant partition $\mathcal{B}$ of $V$
is called a  \textit{block system} for the action of  $G$ on $V$. Each
$B\in\mathcal{B}$  is  called a  \textit{block  of imprimitivity}.  $G$
is called  \textit{imprimitive} in its  action on  $V$ if  it admits  a block
system, otherwise  it is called \textit{primitive}.  Now we suppose that $\Gamma_\infty$ is imprimitive. Let us fix any $k\in\mathcal{K}$ and let $\varphi_k\in\Gamma_\infty$ be the related encryption function. Let $B_1,\dots,B_r$ be a  non-trivial block system for the group $\Gamma_\infty$. This attack works only if we know an efficient algorithm (\emph{block sieving}) with input any vector $v\in V$ and output the (unique) block $B_l$ containing $v$. The classical case is when the block system is formed by all the cosets of a known vector subspace of $V$. Paterson gives this trapdoor for a DES-like cipher (for more details on DES see \cite{specDES}), but it can be extended to the case of AES-like ciphers. We now describe the attack.\\
~\\
~\\
\texttt{Preprocessing performed ones per key}\\ 
We choose one plaintext $m_i$ in each set $B_i$, obtaining the corresponding ciphertext $c_i$. Then the effect of $\varphi_k$ on each block $B_i$ is determined, 
$$
c_i=m_i\varphi_k \in B_j \Rightarrow B_i \varphi_k=B_j.
$$
\texttt{Real-time processing}\\
Given any  ciphertext $c$, we can compute $l$ such that $c\in B_l$ via the block sieving. Then, we can find the plaintext $m$ of $c$ by examining the block $B_l\varphi_k^{-1}$.\\
\texttt{Attack cost}\\
The preprocessing costs $r$ encryptions. For any intercepted ciphertext, the search for the corresponding plaintext is limited to a block, whose size is  $\frac{|V|}{r}$, requiring at most $\frac{|V|}{r}$ encryptions.
~\\

Moreover, a cipher $\mathcal{C}$ may have another weakness if $\Gamma_{\infty}(\mathcal{C})$ is of small size, since not every possible permutation of the message space can be realized by the cipher  \cite{CG,EG83}. Attacks on ciphers whose encryptions generate small groups were given in \cite{KJRS88}. 

In \cite{CVS} the authors define a class of iterated block ciphers, called translation based ciphers (\cite{CVS}, Definition 3.1), large enough to include many common ciphers (as AES \cite{specAES}, SERPENT \cite{specSERPENT} and PRESENT \cite{specPRESENT}), and provide some conditions on the S-Boxes of these ciphers that guarantee the primitivity of $\Gamma_{\infty}$. In particular,  in Theorem 4.4 of \cite{CVS}, it is proved that if $\mathcal{C}$ is a translation based cipher such that any S-Box satisfies, for some integer $r$,
\begin{itemize} 
\item the weak $2^r$-differential uniformity (see Definition \ref{def:weakly} in the next section), and
\item the strongly $r$-anti-invariance (it means that no S-Box of $\mathcal{C}$ sends a proper subspace of codimension $l$ of the plaintext space  to another proper subspace of codimension $l$),
\end{itemize}
then $\Gamma_{\infty}(\mathcal{C})$ is primitive.

\noindent In  Theorem 2 of \cite{onan} ,  under the same hypotheses plus an additional cryptographic assumption, i.e. none of the
images of the derivatives of any S-Box is a coset of a linear subspace of the
message space, it is proved that $\Gamma_{\infty}$ is the alternating group, and so, in other words, $\Gamma_\infty$ is huge.

Starting from these cryptographic motivations, in \cite{onwAPN} the authors provide a deep analysis of 4-bit vectorial Boolean functions focusing on the weak differential uniformity. They determine several conditions, computational and  theoretical, which are either sufficient or necessary for a 4-bit vectorial Boolean function to be weakly $2$-differential uniform (weakly APN). Moreover they consider two non-linearity measures, $\hat{n}(f)$ and $n_i(f)$ where $f$ is a vectorial Boolean function (see  Section \ref{sec:degree}), and they give some relations between such measures and the weakly APNness.

If the image of a derivative of an S-Box of a cipher $\mathcal{C}$ is an affine space then this can be another weakness of $\mathcal{C}$. For example, in the yet unpublished PhD Thesis \cite{phd} the author shows how this condition could induce a weakness based on the action of an alternative operation, called \emph{hidden sum}, for which the  vector space structure of the message space is preserved. In \cite{calsala} some differential properties for permutations in the affine group of the message space  with respect to a hidden sum are investigated. In \cite{ACVS} the authors present conditions on the S-box able to prevent a type of trapdoors based on this alternative operation. One of these conditions is  that the derivatives of the S-box do not map the space to an affine subspace. Moreover, also for the hash function case, in \cite{crooked} the authors show an attack on a SHA-3 candidate (Maraca) \cite{specMARACA}, which is especially effective when the associated vectorial Boolean function has this feature. In Section \ref{sec:powfun} we will show a sufficient condition for monomial weakly APN vectorial Boolean functions in order to have that none of their derivatives sends the message space to a proper affine subspace (Corollary \ref{cor:monomial}).

\section{Weak differential uniformity}\label{sec:1}

Let $\FF=\f2$. Let $m\ge 1$, any vectorial Boolean function (vBf) $f$ from $\FF^m$ to $\FF^m$ can be expressed uniquely as a univariate polynomial in $\fm[x]$. Any time we write that $f$ is a vBf, we will implicit mean $f:\FF^m\to\FF^m$. When $f$ is also invertible we call it a vBf permutation. We denote the \emph{derivative} of $f$ by $\hat{f}_a(x)=f(x+a)+f(x)$, where $a\in\FF^m\setminus\{0\}$, and the \emph{image} of $f$ by $\Im(f) = \{f(x) \mid x \in\FF^m\}$.
\\

A notion of non-linearity for S-Boxes  that has attracted a lot of research.

\begin{definition}
Let $m\ge 1$. Let $f$ be a vBf, for any $a,b\in \FF^m$ we define
$$
\gd_f(a,b)=|\{x\in\FF^m \mid \hat{f}_a(x)=b\}|.
$$
The \emph{differential uniformity} of $f$ is 
$$
\gd(f)=\max_{a,b\in \FF^m\\
a\neq 0}\gd_f(a,b).
$$
$f$ is said $\gd$-\emph{differentially uniform} if $\gd=\gd(f)$.\\
Those functions  with $\gd(f)=2$ are said \emph{Almost Perfect Nonlinear (APN)}.
\end{definition}


There is a  generalization of differential uniformity presented recently in \cite{CVS}, which we recall in the following definition.

\begin{definition}\label{def:weakly}
Let $f$ be a vBf. $f$ is \emph{weakly} $\gd$-\emph{differential uniform} if
$$
|\Im(\hat{f}_a)|>\frac{2^{m-1}}{\gd}, \quad \forall\, a \in\FF^m\setminus\{0\}.
$$
If $f$ is weakly $2$-differential uniform, it is said \emph{weakly APN}.
\end{definition}

As shown in \cite{CVS}, a $\gd$-differentially uniform map is weakly $\gd$-differential uniform. Moreover the following result holds
\begin{lemma}
The weak $\gd$-differential uniformity is affine-invariant.
\end{lemma}
\begin{proof}
If $f$ is weakly $\gd$-differential uniform and  $g(x)=D(f(Cx+c))+d$, for some  $m\times m$ matrices $C$ and $D$ with coefficients in $\FF^m$ and for some $c,d\in\FF^m$, then we have
$$
\begin{aligned}
\hat{g}_a(x)&=D (f(C(x+a)+c))+d+D (f(Cx+c))+d\\
&=D (f(Cx+Ca+c))+D (f(Cx+c))\\
&=D (\hat{f}_{Ca}(Cx+c)),
\end{aligned}
$$
for any $a\in\FF^m$, and so $\Im(\hat{g}_a)=D(\Im(\hat{f}_{Ca}))$.\\ 
Since $C$ and $D$ are permutations, we have that weak $\gd$-differential uniformity is affine-invariant. 
\end{proof}

\begin{remark}
In \cite{bcc} another non-linearity notion, \emph{the locally almost perfect nonlinearity (locally APN)}, is introduced. Note that, in general, the local-APN property is not equivalent to the weak-APN property. For example, the monomial function $x^{11}$ defined over $\FF^6$ is weakly APN but it is not locally APN. However, for any dimension there exist  Boolean functions that are both locally APN and weakly APN, e.g. the patched inversion.
\end{remark}

\begin{remark}\label{rem:2.7}
Suppose that $f$ is not a monomial function. If $f$ is weakly $\gd$-differential uniform then $f^{-1}$ is not necessarily weakly $\gd$-differential uniform. We provide the following example. Let $f:\FF^{4} \to \FF^{4}$ be 
{\small
$$
\begin{aligned}
f(x)&=x^{14}+ e^{10}x^{13}+ex^{12}+e^2x^{11}+e^{9}x^{10}+e^8x^9+e^3x^8+e^5x^7\\
&+e^5x^6+e^{11}x^5+e^8x^3+e^{10}x^2+ex+e^{12},
\end{aligned}
$$
}
where $e$ is a primitive element of $\mathbb{F}_{16}$ such that $e^4=e+1$. The inverse of $f$ is
{\small
$$
\begin{aligned}
f^{-1}(x)&=x^{14}+ e^{10}x^{13}+e^{14}x^{12}+e^8x^{11}+e^{7}x^{10}+e^{10}x^9+x^8+e^5x^7+e^{14}x^6\\
&+e^{2}x^5+e^7x^4+e^5x^3+e^{14}x^2+e^{11}x+e^{14}.
\end{aligned}
$$
}
We have that $f$ is weakly APN while $f^{-1}$ is only weakly 4-differential uniform. 

\end{remark}

We recall that two vBf's $f$ and $g$ are called CCZ-equivalent (Carlet-Charpin-Zinoviev equivalent) if their graphs $G_f=\{(x,f(x)) \mid x\in\FF^m\}$ and $G_{g}=\{(x,g(x)) \mid x\in\FF^m\}$ are affine equivalent, while  they are called EA-equivalent (Extended Affine equivalent) if there exist three affine functions $\lambda$, $\lambda'$ and $\lambda''$ such that $g=\lambda'\circ f\circ \lambda+\lambda''$.

Remark \ref{rem:2.7} and the fact that a vBf $f$ is CCZ-equivalent to $f^{-1}$ imply the following result.

\begin{proposition}
The weak differential uniformity is not CCZ invariant.
\end{proposition} 

On the other hand, weak differential uniformity behaves well with respect to EA invariance, as  shown below.

\begin{proposition}
The weak differential uniformity is EA invariant.
\end{proposition}
\begin{proof}
Let $f$ and $g$ be EA equivalent and let $f$ be weakly $\gd$-differential uniform.\\
Then, $g=\lambda'\circ f\circ \lambda+\lambda''=g'+\lambda''$, with $g'$  affine equivalent to $f$ (and  $\gl''$ is an affinity over $\FF^m$).\\
Since weak differential uniformity is affine invariant, we have $|\Im(\hat{g'}_a)|>2^{m-1}/ \gd$ for all $a \in \FF^m\setminus \{0\}$.\\ 
Note that $\Im(\hat{g}_a)=\{b+\gl''(a) \mid b\in \Im(\hat{g'}_a)\}= \Im(\hat{g'}_a)+\gl''(a)$ and so $|\Im(\hat{g}_a)|=|\Im(\hat{g'}_a)|>2^{m-1}/ \gd$ for any $a \in \FF^m\setminus \{0\}$.
\end{proof}

As seen in Section \ref{sec:0}, if the image of a derivative of an S-Box is an affine space, then there may be a weakness in the cipher. In this direction the following theorem can be useful. Moreover, in Section \ref{sec:powfun} we prove a stronger result regarding the monomial functions (Corollary \ref{cor:monomial}).



\begin{theorem}\label{weaknotAPNcoset}
Let $f$ be a vBf on $\FF^m$ that is weakly APN but not APN. Then, there exists $a\in \FF^m\setminus\{0\}$  such that $\Im(\hat{f}_{a})$ is not a coset of a subspace $W\subseteq \FF^m$. 
\end{theorem}
\begin{proof}

By contradiction suppose that for any $a\neq 0$ we have $ \Im(\hat{f}_{a})=w+ W$ for some $w \in \FF^m$ and $W$ vector subspace of $\FF^m$.
Since $f$ is weakly APN, $|\Im(\hat{f}_{a})|>2^{m-2}$, thus $\dim_{\FF}(W)=m-1$. Therefore, we have that $\hat{f}_a$ is a $2$-to-$1$ function for all $a\neq 0$, which means $f$ is APN, contradicting our assumption.
\end{proof}

\section{Power functions}\label{sec:powfun}

In this section we focus on monomial functions, also called \emph{power functions}. In particular we prove that the weak differential uniformity of a function $f$ is equal to that of $f^{-1}$, and we show some properties  of $\Im(\hat{f}_a)$.

In this section when we write $f=x^d$ we mean that $f$ is a power function on $\fm$ for any $0\leq d\leq 2^{m-1}$. We will also
 identify $\fm$ and $\FF^m$ without any further comments.

The following result is well-known (see for instance \cite{power}).

\begin{lemma}\label{prop:1}
Let $f(x)=x^d$. For any $a,a'\in \FF^m$, with $a,a'\neq 0$, and any $0\le i\le 2^{m}$, we have
$$
|\{b\in\FF^m \mid \gd_f(a',b)=i\}|=|\{b\in\FF^m \mid \gd_f(a,b)=i\}|.
$$
\end{lemma}

\noindent In other words, the differential characteristics of a monomial function depend only on $b$.

\begin{definition}
Let $f(x)=x^d$ and $0\le i\le 2^m$. We denote by $\go_i$ the number of output differences of $b$ that occur $i$ times, that is
$$
\go_i(f)=|\{b \in \FF^m \mid \gd_f(1,b)=i\}|.
$$
The \emph{differential spectrum} of $f$ is the sequence of $\go_i(f)$'s, denoted by $\S(f)$. 
\end{definition}

\begin{remark}
Note that if a monomial function $f$ has $2^s$-to-$1$ derivatives then it is weakly $2^s$-differential uniform, since  $|\Im(\hat{f}_a)|=2^{m-s}$ for any $a\in\FF^m\setminus\{0\}$. 
\end{remark}

The following lemma is well-known, for instance see Lemma 1 in \cite{power}.

\begin{lemma}\label{lm:1}
Let $f(x)=x^d$ and $g(x)=x^e$. If
\begin{itemize}
\item $\gcd(2^m-1,d)=1$ and $ed\equiv {1} \mod (2^m-1)$, 
\item[ ] or
\item $e\equiv 2^kd \mod  (2^m-1)$, for some $k$, $0\leq k\leq m$,
\end{itemize}
then $\quad \S(f)=\S(g)$.

\end{lemma}
From Lemma \ref{lm:1} we obtain our first result.

\begin{corollary}
Let $f(x)=x^d$ with $\gcd(2^m-1,d)=1$. Then $f$ is weakly $\gd$-differential uniform if and only if $f^{-1}$ is weakly $\gd$-differential uniform.
\end{corollary}
\begin{proof}
For a power function we have
$$
|\Im(\hat{f}_a)|=|\Im(\hat{f}_1)| = 2^m-\go_0,\quad\forall a \neq 0.
$$
From Lemma \ref{lm:1} we have $\go_0(f)=\go_0(f^{-1})$, and this concludes the proof. 
\end{proof}

Consider the following lemma for a power function (not necessarily a permutation).

\begin{lemma}\label{prop:2}
Let $f(x)= x^d$. If there exists $a \in \FF^m$, $a\neq 0$, such that $\Im(\hat{f}_a)$ is a coset of a subspace of $\FF^m$, then $\Im(\hat{f}_{a'})$ is a coset of subspace of $\FF^m$ for all $a'\neq 0$.
\end{lemma}
\begin{proof}
We have $\Im(\hat{f}_a)=w+W$, where $W$ is a $\FF$-vector subspace of $\FF^m$ and $w\in \FF^m$.
If we now consider $a'\in \FF^m\setminus\{0\}$ we have 
$$
\hat{f}_{a'}(x)=(x+a')^d+x^d=\left(\frac{a'}{a}\right)^d\left[\left(x\frac{a}{a'}+a\right)^d+\left(x\frac{a}{a'}\right)^d\right]=\left(\frac{a'}{a}\right)^d\hat{f}_a\left(x\frac{a}{a'}\right).
$$
Therefore, $\Im(\hat{f}_{a'})=\left( \frac{a'}{a} \right)^d \Im(\hat{f}_{a})=\left(\frac{a'}{a}\right)^d w+\left(\frac{a'}{a}\right)^dW=w'+W'$.
Since $W'=(a'/a)^d W$, also $W'$ is  an $\FF$-vector subspace of $\FF^m$ and our claim is proved.
\end{proof}

Here we give a sufficient condition for a power function to thwart the aforementioned weakness.\\
The following result is an obvious consequence  of Theorem \ref{weaknotAPNcoset} and Lemma \ref{prop:2}.
\begin{corollary}\label{cor:monomial}
Let $f$ be a vBf permutation on $\FF^m$ that is weakly APN but not APN. If $f(x)=x^d$, then for all $a\in \FF^m\setminus\{0\}$, $\Im(\hat{f}_{a})$ is not a coset of a subspace $W\subseteq \FF^m$.
\end{corollary}

\section{Weakly APN functions and degrees of derivatives}\label{sec:degree}

Without loss of generality, in the sequel we consider only vBf's such that $f(0) = 0$. Let $v\in\FF^m\setminus\{0\}$, we denote by $\langle f,v\rangle$ the component $\sum_{i=1}^m v_if_i$ of $f$, where $f_1,\dots,f_m$ are the coordinate functions of $f$.

\begin{definition}
The \emph{algebraic degree} of a vectorial Boolean function $f$ is the maximal algebraic degree of the coordinate functions of $f$ and is denoted by $\deg(f)$.
\end{definition}

We recall the following non-linearity measures, as introduced in \cite{onwAPN}:
$$
n_i(f):=|\{v\,\in\FF^m \setminus\{0\}\mid\deg(\langle f,v\rangle)=i\}|,
$$
and
$$
\hat n (f):=\max_{a\in\FF^m\setminus\{0\}}|\{v\,\in\FF^m \setminus\{0\}\mid\deg(\langle\hat{f}_a,v\rangle)=0\}|.
$$
\noindent In other words, $n_i(f)$ corresponds to the number of component functions of $f$ which are of
degree $i$ and $\hat n (f)$ corresponds to the number of components of the derivative functions of $f$ which
are constant.
\\
We state two lemmas useful to extend some results of \cite{onwAPN}. First, we recall that the algebraic degree of a permutation $f(x)=x^d$ is the Hamming weight of the binary representation of $d$, denoted by $\mathrm{w}(d)$.
\begin{lemma}[\cite{Kyu07}, Corollary 6]\label{Kyu07}
Let $f(x) = x^d$ be a permutation. Then $\hat{f}_1$ has at least one constant component  if and only if $\mathrm{deg}(f) = 2$.
\end{lemma}
\begin{lemma}[\cite{Her05}, Theorem 1] \label{Her05}
Let $f(x) = x^d$, with $d = 2^{2k} - 2^k + 1$ (Kasami exponent), $\gcd(k, n) = s$ and $\frac{n}{s}$ odd. Then $\hat{f}_1$ is a $2^s$-to-$1$ function.
\end{lemma}

\begin{theorem}\label{th:wapn}

Let $f$ be a vBf permutation such that $\hat{n}(f)=0$, i.e. no component of $f$ has linear structure. Then

\noindent (i) if $m=3$ then $f$ is weakly APN;\\
(ii) if $m=4$ then $f$ is weakly APN;\\
(iii) if $m = 2n$, with $n$ odd, then $f$ is not necessarily weakly APN.
\end{theorem}
\begin{proof}
(i) Let $\FF^3=\{x_1,\dots,x_8\}$ and let $M_a$ be the matrix of dimension $3\times8$, whose columns are $m_j=\hat{f}_a(x_j)$ for $1\le j\le 8$. We claim that $\hat n (f)=0$ implies $\mathrm{rank}(M_a)=3$ for any $a$. Otherwise, we could obtain $(0,\dots,0) \in \FF^8$ from a combination of the rows of $M_a$,
and the corresponding component of $\hat{f}_a$ would be identically $0$.\\
 If $f$ is not weakly APN, we have $|\mathrm{Im}(\hat{f}_a)| \le 2$ for some $a\in\FF^3\setminus{\{0\}}$. So we have at most $2$ distinct columns, which implies $\mathrm{rank}(M_a)\le 2$ and contradicts $\mathrm{rank}(M_a)=3$.
 
 (ii) See \cite{onwAPN} Proposition $2$ .

(iii) Let $t>0$ be such that $\gcd(2^{2^{t+1}} - 2^{2^t} + 1, 2^m - 1) = 1$, and consider the power function $f(x) = x^d$, with $d = 2^{2^{t+1}} - 2^{2^t} + 1$. By hypothesis $\gcd(2^t,m) = 2$, thus, by Lemma \ref{Her05}, $f$ is $4$-differentially uniform and thus weakly $4$-differential uniform. Since in our case $d=2^{2^t}(2^{2^t}-1)+1$, then $\mathrm{w}(d)=\mathrm{w}(2^{2^t}-1)+1$ which is strictly bigger than $2$ for $t>0$. So  $f$ is not quadratic and then, by Lemma \ref{Kyu07}, $\hat{n}(f) = 0$. 
\end{proof}

In \cite{onwAPN} it was shown that a weakly APN function $f$ over $\mathbb{F}^4$ has $n_3(f)\in\{12,14,15\}$, moreover by a computer check on the class representatives the authors exclude the case $n_3(f)=12$ (Fact $4$ in \cite{onwAPN}).\\
We are now able to provide a formal proof.

\begin{proposition}[Fact $4$ in \cite{onwAPN}]
Let $f: \mathbb{F}^4 \to \mathbb{F}^4$ be a weakly APN permutation. Then $n_3(f)\in\{14,15\}$.

\end{proposition}
\begin{proof}
Let $f = (f_1, f_2, f_3, f_4)$ with $f_i : \FF^4\to\FF$, and assume by contradiction that $\deg(S) \le 2$ for three distinct linear combinations $S=\sum_{i}v_if_i$, that we call $S_1, S_2, S_3$. 

From the theory of quadratic Boolean functions (see for instance \cite{quadratic}) $\hat{S}_a$ is constant for every $a \in V(S)$ where $V(S)\subseteq \FF^4 $ is a vector subspace, called the \emph{set of linear structures} of $S$. It is well-known that $V(S)$ has  dimension $0$ if and only if $S$ is bent, it has dimension $4$ if and only if $S$ is linear (affine), and it has dimension $2$ otherwise.  
Since $V(S)$ is a vector space, $S_3=S_1+S_2$.  If there exists $a\in V(S_i)\cap V(S_j)$ different from $0$ for some $i\ne j$, then $\hat{n}(f)\ge 2$. But $f$ weakly APN implies $\hat{n}(f)\le 1$ (see \cite{onwAPN} Theorem $1$). Therefore, $V(S_i)\cap V(S_j)=\{0\}$ and $\dim(V(S_i))\le 2$ for any $i$.\\  
For any $i$, since $f$ is a permutation, then  $S_i$ is balanced, so $S_i$ is not bent, and then 
\begin{equation}\label{S3}
\dim(V(S_i))=2, \quad i=1,2,3 \,.
\end{equation}
Summarizing,  $\deg(S_i)=2$ for any $i$  and $V(S_i)\cap V(S_j)=\{0\}$ for any $i\ne j$.\\
Up to an affine transformtion, since $V(S_1)\oplus V(S_2)=\FF^4$, we can assume 
$V(S_1)=\mathrm{Span}( (1,0,0,0),(0,1,0,0))$ and $V(S_2)=\mathrm{Span}((0,0,1,0),(0,0,0,1))$. 

Let $S_1(x)=\sum_{i<j}c_{i,j}x_ix_j+\sum_i c_i x_i$. Since $S_1(x+(1,0,0,0))+S_1(x)$ is constant we have that $c_{i,j}=0$ if $i$ or $j$ equals $1$. Similarly, since $S_1(x+(0,1,0,0))+S_1(x)$ is constant we have $c_{i,j}=0$ if $i$ or $j$ equals $2$. Then $S_1(x)=x_3x_4+\sum_i c_i x_i$ and analogously we have $S_2(x)=x_1x_2+\sum_i c_i' x_i$, for some $c_i'$'s.\\
So, $S_3(x)=x_1x_2+x_3x_4+\sum_i b_i x_i$, $b_i=c_i+c_i'$, and we can compute the derivative of $S_3$ with respect to $a=(a_1,a_2,a_3,a_4)\in\FF^4\setminus\{0\}$ as
$$
\hat{(S_3)}_a(x)=a_2x_1+a_1x_2+a_4x_3+a_3x_4+c, \mbox{ where $c$ is constant.}
$$
\noindent
Hence $\hat{(S_3)}_a(x)$ is constant if and only if $a=0$, so $S_3$ is bent and  $\dim(V(S_3))=0$, contradicting \eqref{S3}. 
\end{proof}

\section{Quadratic and partially bent functions}\label{sec:part}

\begin{theorem}\label{quadraticweakly}
A quadratic function is APN if and only it is weakly APN.
\end{theorem}
\begin{proof}
Let $f$ be weakly APN and $a\ne 0$ arbitrary. Then by definition, $|\Im(\hat{f}_{a})|>2^{m-2}$. 
Since $f$ is quadratic, $\hat{f}_{a}$ is affine. Then $\Im(\hat{f}_{a})$ is an affine
subspace. Hence its size is a power of $2$, the only possibility of being equal $2^{m-1}$. So $|\Im(\hat{f}_{a})|=2^{m-1}$
 for any non-zero $a$, which means that $\hat{f}_{a}$ is
$2$-to-$1$ for all non-zero $a$.
\end{proof}

As was shown in \cite{SZZ} there is no APN quadratic permutation over $\FF^m$ for $m$ even, and so, by previous theorem, there is no weakly APN quadratic permutation over $\FF^m$ for $m$ even. This result was extended by Nyberg \cite{Ny} to the case of  permutations with all components partially bent  (for $m$ even), in other words there is no APN partially bent permutation. We are able to extend these results by relaxing the condition \emph{APN permutations} with the condition \emph{weakly APN permutation}.

\begin{definition}[\cite{Ca}]
A Boolean function $f$ is partially bent if there exists a linear subspace $V(f)$ of $\FF^m$ such that the restriction of $f$ to $V(f)$ is affine and the restriction of $f$ to any complementary subspace $U$ of $V(f)$, $V(f)\oplus U =\FF^m$, is bent. In that case, $f$ can be represented as a direct sum of the restricted functions, i.e., $f(y + z) = f(y) + f(z)$, for all $z\in V(f)$ and $y\in U$.
\end{definition}

\begin{remark}\label{rk:U}
The space $V(f)$ is formed by the linear structures of $f$, in fact
$$
f(x+a)+f(x)=f(y+z+a)+f(y+z)=f(y)+f(z)+f(a)+f(y)+f(z)=f(a)
$$
where $z,a\in V(f)$ and $y\in U$.
Moreover, since  bent function exist only in even dimension, $m-\dim(V(f))$ is even. That means if $m$ is even, the dimension of $V(f)$ is even. 
\end{remark}

\begin{theorem}\label{th:pb}
For $m$ even, a weakly APN permutation has at most $\frac{2^m-1}{3}$ partially bent components. In particular $f$ cannot have all partially bent components.
\end{theorem}
\begin{proof}
Let $f$ be a weakly APN permutation. Assume by contradiction that $f$ has more than $\frac{2^m-1}{3}$ partially bent components, and denote those with $f_1,\dots,f_s$.
$f$ is a permutation, then $\dim(V(f_i))\neq 0$ for all $1\le i\le s$, otherwise $f_i$ is bent and it is not balanced.
From Remark \ref{rk:U} we have that there exist at least three nonzero vectors in each $V(f_i)$. So
$$
\sum_{i=1}^s |V(f_i)|\ge 3\,s>2^m-1.
$$
Thus, there exist $i$ and $j$ such that $a\in V(f_i)\cap V(f_j)$ with $a\neq 0$. This implies $\hat{n}(f)\ge 2$, which contradicts that $f$ is weakly APN, since in that case $\hat n(f)\le 1$ (\cite{onwAPN} Theorem $1$).
\end{proof}

From the fact that a quadratic Boolean function is partially bent (see for instance \cite{Ny}), we have immediately the following result.

\begin{proposition}

Let $m$ even. Let $f$ be a weakly APN permutation. Then $f$ has at most $2^{m-2}-1$ quadratic components.
\end{proposition}
\begin{proof}
That depends on the fact that the set of components of degree less or equal to $2$ is a vector space.
\end{proof}

\section{Properties linked to $\hat{n}(f)$}
In this last part of the paper we give some properties linked to the value of $\hat{n}(f)$ of a vBf. For all $a\in \FF^m\setminus\{0\}$, let $V_a$ be the vector space $\{v\,\in\FF^m\,:\,\deg(\langle\hat{f}_a,v\rangle)=0\}$. By definition, if $t=\max_{a\in\FF^m\setminus\{0\}}\dim(V_a)$, then $\hat{n}(f)=2^t-1$.
\begin{proposition}
Let $f$ be a vBf and $a\in \FF^m\setminus\{0\}$. $f(a)+V_a^{\perp}$ is the smallest affine subspace of $\FF^m$ containing $\Im(\hat{f}_a)$. In particular, $\hat{n}(f)=0$ if and only if there does not exist a proper affine subspace of $\FF^m$ containing $\Im(\hat{f}_a)$, for all $a\in \FF^m\setminus\{0\}$.
\end{proposition}
\begin{proof}
Let $a\in\FF^m\setminus\{0\}$. Note that $V_a=\{v\,\in\FF^m\,:\,\langle\hat{f}_a,v\rangle \mbox{ is constant} \}$. Let $x\in \FF^m$, then $\hat f_a(x)=f(a)+w$, for some $w\in \FF^m$, and $\langle\hat{f}_a(x),v\rangle=c\in\FF$ for all $v \in V_a$. In particular $c=\langle\hat{f}_a(0),v\rangle\,=\,\langle f(a),v\rangle$ and so $\langle w,v\rangle=0$, that is, $w\in V_a^\perp$. Then we have $\Im(\hat f_a)\subseteq f(a)+V_a^\perp$. Now, let $A$ be an affine subspace containing $\Im(\hat{f}_a)$, then $A=f(a)+V$, for some vector subspace $V$ in $\FF^m$. For all $v\in V^{\perp}$, we have $\langle\hat{f}_a,v\rangle=\langle f(a),v\rangle=c\in\FF$ and so, by definition, $V^{\perp}\subseteq V_a$. Then $A$ contains $f(a)+V_a^{\perp}$.

Finally,  $\hat{n}(f)=0$ if and only if $V_a=\{0\}$ for all $a\in\FF^m\setminus\{0\}$, and so our claim follows.
\end{proof}

\begin{remark}
The proposition above gives a sufficient condition, \emph{i.e.} $\hat{n}(f)=0$, such that the derivates do not map the message space to an affine subspace; and so a type of trapdoors can be avoided, as noted in Section \ref{sec:1}. 
\end{remark}

The following proposition is well-known, but we propose a proof in our context.

\begin{proposition}\label{prop:15}
Let $f: \mathbb{F}^m \to \mathbb{F}^m$ be a Boolean permutation such that $\hat{n}(f)=0$. Then $f$ has no partially bent (quadratic) components.
\end{proposition}
\begin{proof}
$\hat{n}(f)=0$ implies that the linear structures set of any component contains only $0$. So if there exists a partially bent (quadratic) component, then it is bent. But $f$ is a permutation, then this is not possible.
\end{proof}

For the particular case of $4$-bit S-Boxes we obtain two more results.

\begin{corollary}
Let $f: \mathbb{F}^4 \to \mathbb{F}^4$ be a vBf permutation.\\
(i) If $\hat{n}(f)=0$. Then  $f$ is weakly APN and $n_3(f)=15$.\\
(ii) If $f$ is weakly APN and $n_3(f)=14$. Then $\hat{n}(f)=1$.

\end{corollary}
\begin{proof}
Let $f$ be weakly APN, so $\hat{n}(f)\leq 1$ (see \cite{onwAPN}). From Proposition \ref{prop:15}, the claim follows.
\end{proof}

So for weakly APN function $f:\FF^4\to \FF^4$ we have all the three cases. Below we provide an example for each case reporting the algebraic normal form of the components of $f$: 
\begin{itemize}

\item $\hat{n}(f)=0$ and $n_3(f)=15$:
$$
\begin{aligned}[l]
&f_1=x_1x_2x_3+ x_2x_3x_4 + x_1x_3  + x_2x_3+ x_1  + x_2 + x_3 +  x_4\\
&f_2=x_1x_2x_4 + x_1x_2 + x_1x_3 + x_2x_3 + x_2x_4 + x_4\\
&f_3=x_1x_3x_4 +x_1x_2 +  x_1x_3 + x_1x_4 + x_3 + x_4\\
&f_4=x_2x_3x_4 +x_1x_4 +  x_2x_4 + x_2 + x_3x_4 + x_3 + x_4
\end{aligned}
$$
\item $\hat{n}(f)=1$ and $n_3(f)=15$:
$$
\begin{aligned}[l]
&f_1=x_1x_3x_4  + x_2x_3x_4 + x_2x_3 + x_2x_4 + x_3x_4+ x_1\\
&f_2=x_1x_2x_4 + x_1x_3 + x_1x_4 + x_2x_3 + x_2\\
&f_3=x_1x_2x_3 + x_1x_2x_4+ x_1x_3x_4 + x_2x_3x_4  + x_1x_2 +x_3x_4 + x_3\\
&f_4=x_2x_3x_4 +x_1x_2 + x_1x_4 +  x_2x_3 + x_4
\end{aligned}
$$
\item $\hat{n}(f)=1$ and $n_3(f)=14$:
$$
\begin{aligned}[l]
&f_1=x_1x_2x_3 + x_1x_2x_4 + x_1x_3 + x_1 + x_2x_3x_4 + x_2x_3 + x_3x_4\\
&f_2=x_1x_2x_4 + x_1x_2 + x_1x_3x_4 + x_1x_3 + x_1x_4 + x_2\\
&f_3=x_1x_2x_4 + x_1x_2 + x_1x_3x_4 + x_1x_3 + x_2x_4 + x_3\\
&f_4=x_1x_3 + x_1x_4 + x_2x_3x_4 + x_2x_4 + x_4
\end{aligned}
$$
\end{itemize}

\section{Conclusions}
As reported in Section \ref{sec:0} and \ref{sec:1}, weak differential uniformity and the cryptographic condition that the image of the derivatives of an S-Box is never a coset of a subspace of the message space are useful to prevent  hiding certain type of trapdoors in the related cipher. 

First we study the algebraic structure of the image of the derivatives of a vectorial Boolean function. In particular we prove that for any vBf that is weakly APN but not APN, there is at least one derivative whose image is not an affine subspace (Theorem \ref{weaknotAPNcoset}). In the case of power functions, to be weakly APN but not APN guarantees that none of the image of the derivatives is an affine subspace (Corollary \ref{cor:monomial}). An interesting open problem is to find families of  vBf that are not monomial  having this property.

Then we show that the sufficient condition $\hat{n}(f)=0$, that ensures weakly APNness for the 4-bit vBf's, does not guarantee this property for $m$-bit vBf's with $m>4$ (Theorem \ref{th:wapn}). It would be interesting to find sufficient conditions that imply weakly APNness for any $m>4$.

In Section  \ref{sec:part} we extend some results known for the (quadratic) partially bent components of an APN permutation to the case of weakly APN permutations.

In the last section we report some other results linked to the value of $\hat{n}(f)$, in particular we prove that this value is zero if and only if the derivates of $f$ do not map the message space to an affine subspace.

\section*{Acknowledgements}
We are grateful to the unknown referees for their suggestions, which were decisive in order to improve and clarify the exposition. In particular we would like to thank one of the referees for the  Theorem \ref{quadraticweakly} and its proof.

\end{document}